\documentclass[11pt]{article} 

\usepackage{color}
\usepackage{amsfonts}
\usepackage{amsmath}
\usepackage{amsthm}
\usepackage{scalefnt}
\usepackage{hyperref}
\usepackage{authblk}
\usepackage{mathtools}
\usepackage{enumitem}

\usepackage{graphicx}
\usepackage[font=small,labelfont=bf]{caption}
\usepackage[a4paper, total={6in, 9in}]{geometry}

\newlength{\bibitemsep}
\newlength{\bibparskip}
\let\oldthebibliography\thebibliography
\renewcommand\thebibliography[1]{%
  \oldthebibliography{#1}%
  \setlength{\parskip}{\bibitemsep}%
  \setlength{\itemsep}{\bibparskip}%
}

\newtheorem{theorem}{Theorem}

\newtheorem{corollary}[theorem]{Corollary}

\newtheorem{lemma}[theorem]{Lemma}

\newtheorem{remark}[theorem]{Remark}
\newtheorem{observation}[theorem]{Observation}

\newcommand{\fan}{fan}

\newcommand{\orgadget}{\ensuremath{\bigcirc}\kern-0.321cm\raisebox{-.022cm}{\ensuremath{\vee}}\kern.05cm}
\newcommand{\forceedge}{\ensuremath{\bigcirc}\kern-0.28cm\raisebox{-.0001cm}{\ensuremath{\uparrow}}\kern.08cm}
\newcommand{\forceedgesmall}{\ensuremath{\bigcirc}\kern-0.26cm\raisebox{-.0001cm}{\ensuremath{\uparrow}}\kern.08cm}
\newcommand{\xorgadget}{$\bigcirc$\kern-0.31cm\raisebox{-.0cm}{$\vee$}\kern.05cm}

\title{Recognition and Complexity of Point Visibility Graphs\thanks{This manuscript is a revised and extended version of a previous paper published in the proceedings of the 31st International Symposium on Computational Geometry (SoCG 2015). J. Cardinal is supported by the ARC (Action de Recherche Concert\'ee) project COPHYMA. U. Hoffmann is supported by the Deutsche Forschungsgemeinschaft within the research training group 'Methods for Discrete Structures' (GRK 1408).}}

\author{Jean Cardinal\ \ \ \ \ \ Udo Hoffmann\\
Universit\'e libre de Bruxelles (ULB), Brussels, Belgium}
\date{September 2016}
\begin{document}
\thispagestyle{empty}
\maketitle

\begin{abstract}
 A {\em point visibility graph} is a graph induced by a set of points in the plane,
 where every vertex corresponds to a point, and two vertices are adjacent whenever
 the two corresponding points are {\em visible} from each other, that is,
 the open segment between them does not contain any other point of the set. 

We study the recognition problem for point visibility graphs: given a simple undirected graph, decide whether it is the visibility graph of some point set in the plane. We show that the problem is complete for the existential theory of the reals. Hence the problem is as hard as deciding the existence of a real solution to a system of polynomial inequalities. The proof involves simple substructures forcing collinearities in all realizations of some visibility graphs, which are applied to the algebraic universality constructions of Mn\"ev and Richter-Gebert. This solves a longstanding open question and paves the way for the analysis of other classes of visibility graphs.

Furthermore, as a corollary of one of our construction, we show that there exist point visibility graphs
that do not admit any geometric realization with points on a grid. We also prove that the problem of recognizing
visibility graphs of points on a grid is decidable if and only if the existential theory of the rationals is decidable.
\end{abstract}

\section{Introduction}

Visibility between geometric objects is a cornerstone notion in discrete and computational geometry,
that appeared as soon as the late 1960s in pioneering experiments in robotics~\cite{LP79}.  
Visibility is involved in major themes that helped shape the field, such as art gallery and motion
planning problems~\cite{dutch,ghosh,artgallery}. 
However, despite decades of research on those topics, the combinatorial structures
induced by visibility relations in the plane are far from understood. 

Among such structures, {\em visibility graphs} are arguably the most natural.
In general, a visibility graph encodes the binary, symmetric visibility relation among sets of objects in the plane,
where two objects are visible from each other whenever there exists a straight line of sight between them
that does not meet any obstacle. More precisely, a {\em point visibility graph} associated 
with a set $P$ of points in the plane is a 
simple undirected graph $G=(P,E)$ such that two points of $P$ are adjacent if and only if
the open segment between them does not contain any other point of $P$. Note that the points play
both the roles of vertices of the graph and obstacles. In what follows, we will use the abbreviation
PVG for point visibility graph.

We consider the {\em recognition} problem for point visibility graphs: given a simple undirected graph
$G=(V,E)$, does there exists a point set $P$ such that $G$ is isomorphic to the visibility graph of $P$?
More concisely, the problem consists of deciding the property of being a point visibility graph of some
point set. As is often the case for geometric graphs, the recognition problem appears to be intractable under usual 
complexity-theoretic assumptions. In fact, recently, Roy gave an NP-hardness proof for this problem~\cite{roy2014point}.

\subsection{Our results}

We characterize the problem as complete for the 
existential theory of the reals; hence recognizing point visibility graphs is
as hard as deciding the existence of a solution to an arbitrary system of polynomial inequalities
over the reals. Equivalently, this amounts to deciding the emptiness of a semialgebraic set.
This complexity class is intimately related to fundamental results on {\em oriented matroids} and
{\em pseudoline arrangements} starting with the insights of Mn\"ev on the algebraic universality properties
of these structures~\cite{mnev1988universality}.
The notation $\exists\mathbb{R}$ has been proposed recently by Schaefer~\cite{S09} to refer to this class, motivated by the continuously 
expanding collection of problems in computational
geometry that are identified as complete for it.

The only known inclusion relations for $\exists\mathbb{R}$
are $NP\subseteq\exists\mathbb{R}\subseteq PSPACE$. It is known from the Tarski-Seidenberg Theorem that 
the first-order theory of real closed fields is decidable, but polynomial space algorithms for
problems in $\exists\mathbb{R}$ have been proposed only much more recently by Canny~\cite{canny1988}.

Whenever a graph is known to be a point visibility graph, the description of the point set 
as a collection of pairs of integer coordinates constitutes a natural certificate. 
Since it is not known whether $\exists\mathbb{R}\subseteq NP$, we should
not expect such a certificate to have polynomial size. 
In fact, we show that there exist point visibility graphs all realizations
of which have an irrational coordinate, and point visibility graphs that require
doubly exponential coordinates in any realization. We also prove that recognizing 
visibility graphs of points on a grid is decidable if and only if the existential theory 
of the rationals is decidable. This establishes an interesting connection between a natural graph
drawing problem and Hilbert's tenth problem over the rationals.

\subsection{Related work and Connections}

The recognition problem for point visibility graphs has been explicitly stated as an important 
open problem by various authors~\cite{KPW05}, and is listed as the first open problem in a recent survey
from Ghosh and Goswami~\cite{ghosh2013unsolved}.

A linear-time recognition algorithm has been proposed by Ghosh and Roy for {\em planar}
point visibility graphs~\cite{GR14}. For general point visibility graphs they showed that recognition problem lies in $\exists\mathbb{R}$.
More recently, Roy~\cite{roy2014point} published an ingenious
and rather involved NP-hardness proof for recognition of arbitrary point visibility graphs. 
Our result clearly implies NP-hardness as well, and, in our opinion, has a more concise proof.

Structural aspects of point visibility graphs have been studied by K\'ara, P\'or, and Wood~\cite{KPW05}, 
P\'or and Wood~\cite{PW10}, and Payne et al.~\cite{PPVW12}. Many fascinating open questions revolve around
the {\em big-line-big-clique} conjecture, stating that for all $k,\ell\geq 2$, there exists
an $n$ such that every finite set of at least $n$ points in the plane contains either $k$ pairwise
visible points or $\ell$ collinear points.

{\em Visibility graphs of polygons} are defined over the vertices of an arbitrary simple polygon in the plane, 
and connect pairs of vertices such that the open segment between them is completely contained in the interior
of the polygon. This definition has also attracted a lot of interest in the past twenty years. 
Ghosh gave simple properties of visibility graphs of polygons and conjectured
that they were sufficient to characterize
visibility graphs~\cite{G97}. These conjectures have been disproved by Streinu~\cite{S05}
via the notion of {\em pseudo-visibility} graphs, or visibility graphs of {\em pseudo-polygons}~\cite{ORS97}. 
A similar definition is given by Abello and Kumar~\cite{AK02}.
Roughly speaking, the relation between visibility and pseudo-visibility
graphs is of the same nature as that between arrangements of straight lines and pseudolines. 
Although, as Abello and Kumar remark, these results somehow suggest that the difficulty 
in the recognition task is due to a stretchability problem, the complexity of recognizing visibility graphs 
of polygons remains open, and it is not clear whether the techniques described in this paper can help characterizing it.
The influential surveys and contributions of Schaefer about $\exists\mathbb{R}$-complete problems in
computational geometry form an ideal point of entry in the field~\cite{S09,S12}.
Among such problems, let us mention
recognition of segment intersection graphs~\cite{kratochvil1994intersection}, recognition of unit distance graphs and realizability
of linkages~\cite{K02,S12}, recognition of disk and unit disk intersection graphs~\cite{MM13},
computing the rectilinear crossing number of a graph~\cite{B91}, simultaneous geometric graph embedding~\cite{kyncl2011simple}, 
and recognition of $d$-dimensional Delaunay triangulations~\cite{APT15}.

\subsection{Outline of the paper}

In Section~\ref{sec:uniquerepresentations}, we provide a simple visibility graph construction, 
a {\em fan}, all geometric realizations of which are guaranteed
to preserve a specified collection of subsets of collinear points. 
The proofs are elementary and only require a series of basic observations.

The main result of the paper is given in Section~\ref{sec:reduction}. We first recall the main notions and tools
used in the results from Mn\"ev~\cite{mnev1988universality} and
Shor~\cite{shor1991stretchability} for showing that realizability of abstract 
order types is complete for the existential theory of the reals.
We then combine these tools with the fan construction to produce families of point visibility
graphs that can simulate arbitrary arithmetic computations over the reals.

In Section~\ref{sec:grid}, we give two applications of the fan construction.
In the first, we show that there exists a point visibility graph that does not have any geometric
realization on the integer grid.
In other words, all geometric realizations of this point visibility graph are
such that at least one of the points has an irrational coordinate.
Another application of the fan construction follows,
where we show that there are point visibility graphs every grid
realization of which requires coordinates of values $2^{2^{\sqrt[3]{n}}}$
where $n$ denotes the number of vertices of the PVG.
Afterwards, we show that the recognition of visibility graphs of points on a grid (or, equivalently, with integer coordinates)
might be undecidable by reducing from the solvability of a system of polynomial (in)equalities over the rationals.

\subsection{Notation}

For the sake of simplicity, we slightly abuse notation and 
do not distinguish between a vertex of a point visibility graph and its 
corresponding point in a geometric realization.
We denote by $G[P']$ the induced subgraph of a graph $G=(P,E)$ with the vertex set $P'\subseteq P$.
For a point visibility realization $R$ we denote by $R[P']$ the induced subrealization containing only
the points $P'$. The PVG of this subrealization is in general not an induced subgraph of $G$.
By $N(p)$ we denote the open neighborhood of a vertex $p$.

The line through two points $p$ and $q$ is denoted by $\ell(p,q)$
and the open segment between $p$ and $q$ by $\overline{pq}$.
We will often call $\overline{pq}$ the \emph{sightline} between $p$ and $q$,
since $p$ and $q$ see each other iff $\overline{pq}\cap P=\emptyset$.
We call two sightlines $\overline{p_1q_1}$ and $\overline{p_2q_2}$ \emph{non-crossing} if
$\overline{p_1q_1}\cap\overline{p_2q_2}=\emptyset$.

For each point $p$ all other points of $G$ lie on $\deg(p)$ many rays
$R^p_1,\dots,R^p_{\deg(p)}$ originating from $p$.

\section{Point visibility graphs preserving collinearities}\label{sec:uniquerepresentations}

We first describe constructions of point visibility graphs, all the geometric realizations of which  
preserve some fixed subsets of collinear points.

\subsection{Preliminary observations}

\begin{figure}[ht]
\centering
 \includegraphics[width=.6\textwidth]{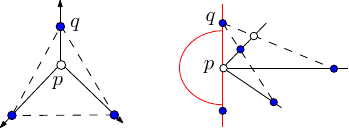}
 \caption{\label{fig:stableline_emptyhalfspace} (Lemma~\ref{lem:empty_halfspace}) Left: a point sees points on consecutive rays with small angle.
Right: a vertex $q$ of degree one in $G[N(p)]$ lies on the boundary of an empty halfspace.}
\end{figure}

In the realization of a PVG, the point $p$ sees exactly $\deg(p)$ many vertices,
hence all other points lie on $\deg(p)$ rays of origin $p$.

\begin{lemma}\label{lem:empty_halfspace}
  Let $q\in N(p)$ be a degree-one vertex in $G[N(p)]$.
  Then all points lie on one side of the line $\ell (p,q)$.
\end{lemma}
\begin{proof}
  If the angle between two consecutive rays is smaller than $\pi$, then
  every vertex on one ray sees every vertex on the other ray.
  Hence one of the angles between two rays of origin $p$ must be at least $\pi$ (see Figure~\ref{fig:stableline_emptyhalfspace}).
\end{proof}

\begin{corollary}\label{cor:pathUnique}
  If $G[N(p)]$ is an induced path, then the order of the path and the order of the rays coincide.
\end{corollary}
\begin{proof}
  By Lemma~\ref{lem:empty_halfspace} the two endpoints of the path lie on
  rays on the boundary of empty halfspaces. Thus all other rays form angles
  which are smaller than $\pi$, and thus they see their two neighbors
  of the path on their neighboring rays.
\end{proof}
\begin{observation}\label{obs:secondPoint}
  Let $q$, $q\not=p$, be a point that sees all points of $N(p)$.
  Then $q$ is the second point (not including $p$) on one of the rays emerging from $p$.
\end{observation}
\begin{proof}
  Assume $q$ is not the second point on one of the rays. 
  Then $q$ cannot see the first point on its ray which is a neighbor of $p$.
\end{proof}

This also yields the following observation.

\begin{observation}\label{obs:nonSecondPoint}
  Let $q$, $q\not= p$, be a point that is not the second point on one of the rays from $p$
  and sees all but one of the neighbors of $p$, say $r$. 
  Then $q$ lies on the ray of $r$.
\end{observation}

\subsection{Fans}
\begin{figure}[ht]
\centering
 \includegraphics[scale=.7]{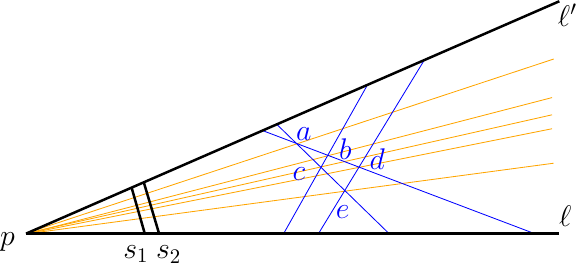}
 \caption{\label{fig:fan} A fan: a vertex is placed on each intersection of two lines/segments.}
\end{figure}

We have enough tools by now to show the uniqueness of a PVG obtained from the following construction, which is depicted in Figure~\ref{fig:fan}.
Consider a set $S$ of segments between
two lines $\ell$ and $\ell'$ intersecting in a point $p$,
such that each endpoint of a segment lies on $\ell$, $\ell'$ or another
segment. 
For each intersection of a pair of segments, construct a ray of origin $p$ and
going through this intersection point.
Add two segments $s_1$ and $s_2$ between $\ell$ and $\ell'$, such
that $s_1$ is the closest and $s_2$ is the second closest segments to~$p$.

We now put a point on each intersection of the segments and rays and construct the PVG of this set
of points. We call this graph the \emph{fan} of $S$ and denote it by $\fan(S)$.
Since we have the choice of the position of the segments $s_1$ and $s_2$
we can avoid any collinearity between a point on $s_1$ or $s_2$ and points on other segments,
except for the obvious collinearities on one ray. 
Thus every point sees all points on $s_1$ except for the one of the ray it lies on.

\begin{lemma}\label{lem:uniqueFan}
All realizations of a fan preserve collinearities between points that lie on one segment and between points that lie on one ray.
\end{lemma}
\begin{proof}
  We first show that the distribution of the points onto the rays of $p$ is unique.
  By construction, the points on $s_2$ see all the points on $s_1$, which are exactly the neighbors of $p$.
  Thus by Observation~\ref{obs:secondPoint} the points from $s_2$ are the second points of a ray.
  Since there is exactly one point for each ray on $s_2$, all the other points are not second points on a ray.
  By construction, each of the remaining points sees all but one point of $s_1$.
  Observation~\ref{obs:nonSecondPoint} gives a unique ray a point lies on.
  The order of the rays is unique by Corollary~\ref{cor:pathUnique}.
  On each ray the order of the points is as constructed, since the PVG of points on one ray 
  is an induced path.

  Now we have to show that the points originating from one segment are still collinear.
  Consider three consecutive rays $R_1,R_2,R_3$. 
  We consider a visibility between a point $p_1$ on $R_1$ and one point $p_3$ on $R_3$
  that has to be blocked by a point on $R_2$.
  Let $p_2$ be the original blocker from the construction.

  For each point $q$ on $R_2$ that lies closer to $p$ than $p_2$ there is a sightline blocked by $q$,
  and for each point $q$ that lies further away from $p$ than $p_2$ there is a sightline blocked by $q$.
  For each point $q$ on $R_2$ we pick one sightline between a point on $R_1$ and another point on $R_3$ that is blocked by $q$.
  This set of sightlines is non-crossing, since the segments only intersect on rays, hence on $R_2$.
  So we have a set of non-crossing sightlines and the same number of blockers available.
  Since the order on each ray is fixed, and the sightlines intersect $R_2$ in a certain order,
  the blocker for each sightline is uniquely determined and has to be the original blocker.
  By transitivity of collinearity all points from the segments remain collinear over all the rays.
\end{proof}

\section{$\exists\mathbb{R}$-completeness reductions and universality}\label{sec:reduction}

The existential theory of the reals ($\exists\mathbb{R}$) is a complexity
class defined by the following complete problem. 
We are given a well-formed quantifier-free formula $F(x_1,\dots,x_k)$ 
using the numbers $0$ and $1$, addition and multiplication operations, 
strict and non-strict comparison operators,
Boolean operators, and the variables $x_1,\dots,x_k$,
and we are asked whether there exists an assignment of real values to $x_1,\dots,x_k$,
such that $F$ is satisfied. This amounts to deciding whether a system of polynomial inequalities admits
a solution over the reals.
The first main result connecting this complexity class to a geometric 
problem is the celebrated result of Mn\"ev, who showed that \emph{realizability of order types},
or -- in the dual -- stretchability of pseudoline arrangements,
is complete in this complexity class~\cite{mnev1988universality}.
In what follows, we use the simplified reduction due to
Shor~\cite{shor1991stretchability}.

The \emph{orientation} of an ordered triple of points $(p,q,r)$ indicates whether the three points form
a clockwise or a counterclockwise cycle, or whether the three points are collinear.
Let  $P=\{p_1,\dots,p_n\}$ and an orientation $O$ of each triple of points in $P$ be given.
The pair $(P,O)$ is called an \emph{(abstract) order type}.
We say that the order type $(P,O)$ is realizable if there are coordinates in the plane for the points of $P$,
such that the orientations of the triples of points match those prescribed by $O$.

In order to reduce the order type realizability problem to solvability of
a system of strict polynomial inequalities, we have to be able to simulate
arithmetic operations with order types. This uses standard constructions introduced by
von Staudt in his ``{\em algebra of throws}''~\cite{staudt}.

\subsection{Arithmetics with order types.\label{subsec:arithmetics}}

To carry out arithmetic operations using orientation predicates, we
associate numbers with points on a line, and use the \emph{cross-ratio} to encode their values.

The cross ratio $(a,b;c,d)$ of four points $a,b,c,d\in\mathbb{R}^2$ is defined as
$$(a,b;c,d):=\frac{|a,c|\cdot|b,d|}{|a,d|\cdot|b,c|},$$
where $|x,y|$ is the determinant of the matrix obtained by writing the two vectors as columns.
The two properties that are useful for our purpose is that the cross-ratio
is invariant under projective transformations, and that for four points on one line, the cross-ratio
is given by $\frac{\overrightarrow{ac}\cdot\overrightarrow{bd}}{\overrightarrow{ad}\cdot\overrightarrow{bc}}$,
where $\overrightarrow{xy}$ denotes the oriented distance between $x$ and $y$ on the line.

We will use the cross-ratio the following way:
We fix two points on a line and call them $0$ and $1$.
On the line through those points we call the point at infinity $\infty$.
For a point $a$ on this line the cross-ratio $x:=(a,1;0,\infty)$ results in the distance between $0$ and $a$
scaled by the distance between $0$ and $1$. 
Because the cross-ratio is a projective invariant we can fix one line
and use the point $a$ for representing the value $x$.
In this way, we have established the coordinates on one line.

\begin{figure}[ht]
\centering
 \includegraphics[width=.9\textwidth]{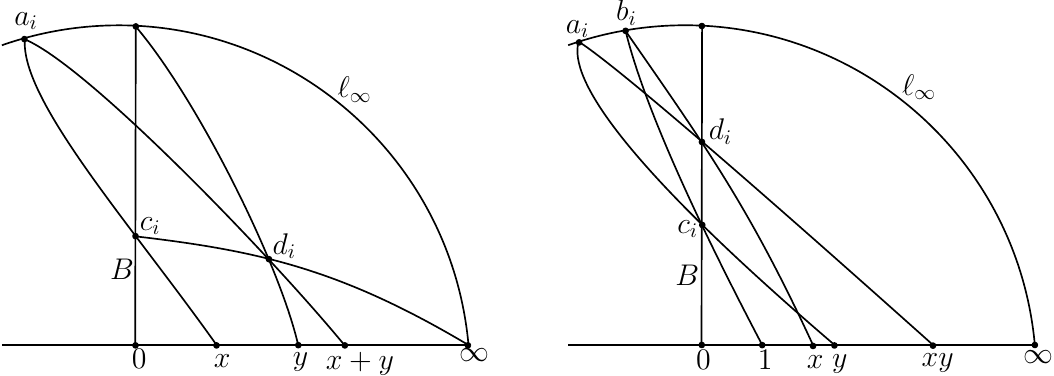}
 \caption{\label{fig:gadgets} Gadgets for addition (left) and multiplication (right) on a line.}
\end{figure}

For computing on this line, the gadgets for addition and multiplication
depicted in Figure~\ref{fig:gadgets} can be used.
Let us detail the case of multiplication.
We are given the points $0< 1< x< y<\infty$ on the line $\ell$, and wish
to construct a point on $\ell$ that represents the value $x\cdot y$.
Take a second line $\ell_\infty$ that intersects $\ell$ in $\infty$,
and two points $a,b$ on this line. Construct the segments $\overline{by},\overline{b1}$ and $\overline{ax}$.
Denote the intersection point of $\overline{ax}$ and $\overline{b1}$ by $c$.
Call $d$ the intersection point of $\overline{by}$ and $\ell(0,c)$.
The intersection point of $\ell$ and $\ell(d,a)$ represents the point $x\cdot y=:z$ on $\ell$,
i.e., $(z,1;0,\infty)=(x,1;0,\infty)\cdot (y,1;0,\infty)$.
In a projective realization of the gadget in which the line $\ell_{\infty}$ is indeed the line
at infinity, the result can be obtained by applying twice the intercept theorem, in the triangles 
with vertices $0, d, y$ and $0, d, z$, respectively.
To add the cross ratios of two points on a line, a similar construction is given in Figure~\ref{fig:gadgets}.

\subsection{The reduction for order types}

Using the constructions above we can already model a system of strict polynomial
inequalities. However, it is not clear how we can determine the complete
order type of the points without knowing the solution of the system. Circumventing this obstacle 
was the main achievement of Mn\"ev~\cite{mnev1988universality}. 
We cite one of the main theorems in a simplified version.

\begin{theorem}[\cite{shor1991stretchability}]
  Every \emph{primary semialgebraic set} $V\subseteq\mathbb{R}^d$ is \emph{stably equivalent}
  to a semialgebraic set $V'\subseteq\mathbb{R}^n$, with $n=\mathrm{poly}(d)$, 
  for which all defining equations have the form 
  $x_i+x_j=x_k$ or $x_i\cdot x_j=x_k$ for certain $1 \leq i \leq j < k\leq n$, where the
  variables $1=x_1<x_2<\dots<x_n$ are totally ordered.
\end{theorem}

A \emph{primary semialgebraic set} is a set defined by polynomial equations and strict polynomial inequalities
with coefficients in $\mathbb{Z}$.
Although we cannot give a complete definition of {\em stable equivalence} within the context of this paper,
let us just say that two semialgebraic sets $V$ and $V'$ are stably equivalent
if one can be obtained from the other by rational transformations and so-called {\em stable projections},
and that stable equivalence implies {\em homotopy equivalence}. 
From the computational point of view, the important property is
that $V$ is the empty set if and only $V'$ is, and that the size of the description of $V'$ in the
theorem above is polynomial in the size of the description of $V$. 
In the proof of universality for PVGs we will only use that we construct PVGs
that contain a subset of points whose order type is the one constructed by
Shor, which has a certain wanted realization
space.
We call the description of a semialgebraic set $V'$ given in the theorem above the \emph{Shor normal form}.

We can now encode the defining relations of a semialgebraic set given in Shor normal form using
abstract order types by simply putting the points $0,1,x_1,\dots,x_n,\infty$ in this order on $\ell$.
To give a complete order type, the orientations of triples including the points of the gadgets and the
positions of the gadget on $\ell_\infty$ have to be specified.
This can be done exploiting the fact that the distances between the points
$a_i$ and $b_i$ of each gadget
and their position on $\ell_\infty$ can be chosen freely. We refer to the references mentioned above for 
further details. We next show how to implement these ideas to construct a graph
$G_V$ associated with a primary semialgebraic set $V$, such that $G_V$ has a PVG realization
if and only if $V\not=\emptyset$.

\section{$\exists\mathbb{R}$-completeness of PVG recognition}

The idea to show that PVG recognition is complete in $\exists\mathbb{R}$
is to encode the gadgets described in the previous section in a fan.

We therefore consider the gadgets not as a collection of points with given order types,
but as a arrangement of segments.
This arrangement can be fixed in a fan if the radial ordering around $p$, the origin of the
fan, is known.

We will consider the addition and multiplication gadgets given in Fig.~\ref{fig:gadgets},
and for a copy $g_i$ of the addition gadget, denote by $a_i,b_i,c_i$, and $d_i$ the
points corresponding to $g_i$, and similarly for the multiplication gadget. 
We describe how to place these points, such that we are able to describe the
complete order type. In addition, we are able to fix the order of the
$x$-coordinates of the intersection points of the arrangement, such that it
does not restrict realizability. This allows us to place $p$ at $(0,-C)$ for
a large number $C$, such that the order of the $x$-coordinates of intersection points agrees with the radial
ordering around $p$.

\begin{theorem}\label{thm:PVGreduction}
The recognition of point visibility graphs is $\exists\mathbb{R}$-complete.
\end{theorem}

\begin{proof}
  Given a primal semialgebraic set $V$, we construct a graph whose realization space as
  point visibility graph is stably equivalent to $V$.
  For a primary semialgebraic set $V$ we compute the Shor normal form and denote the
  corresponding primary semialgebra set by $V'$.
  For $V'$, we construct the order type that implements the calculations
  on the lines $\ell$ and $\ell_\infty$ using the construction of Shor~\cite{shor1991stretchability}:
  we iteratively place the points $a_i, b_i$ of the gadgets on
  $\ell_\infty$ and the points $c_i$ and $d_i$ on a vertical segment $B$ that
  starts on $\ell$.
  The points of $g_i$ are placed closer to the intersection point of $B$ and
  $\ell_\infty$ at each step. This can be done in a way that allows us to
  determine a possible order of the $x$-coordinates of all intersection points
  of the segments constructed.
  First, using a projective transformation we can assume that $\ell_\infty$ is indeed
  the line at infinity. Then we place a new gadget ``close'' to the intersection
  point of $B$ and $\ell_\infty$. This corresponds to choosing segments with a higher
  absolute values for the slopes.
  We choose the slope of a new segment $s$ large enough, so that
  all intersection points with the segments constructed before lie in an
  interval just to the left of the point on $\ell$ that $s$ is originating
  from. This interval to the left is indicated by the gray box in
  Figure~\ref{fig:order_shor} (left).
  
\begin{figure}[ht]
\centering
  \includegraphics[width=.9\textwidth]{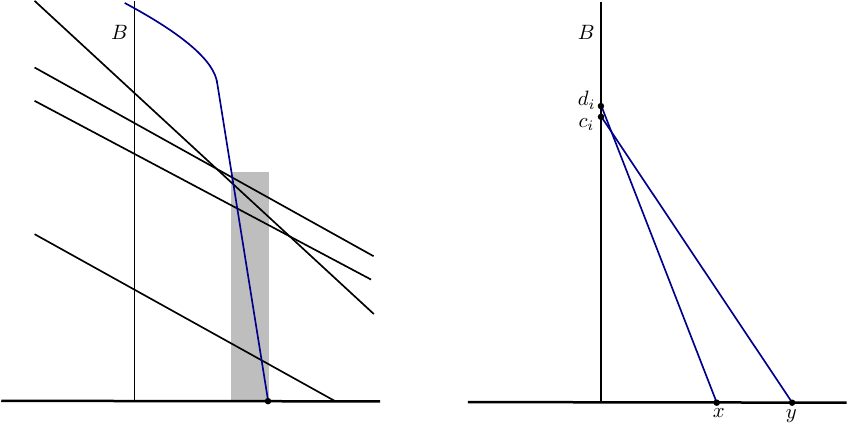}
 \caption{\label{fig:order_shor} Left: The intersection points of a new segment
  (blue) with the old segments (black) can be forced in an arbitrary small
  interval left to the origin of the new segment on $\ell$. Right: The
  intersection point on the left of $B$ in a multiplication gadget can be
  forced to lie arbitrary close to $B$.}
\end{figure}

  So it only remains to determine the relative $x$-position of the intersection points of
  segments within one gadget among all intersection points.  
  The intersection points of the segments in a multiplication gadget are the
  closest points to the vertical segment $B$.
  This is achieved by constructing $c_i$ and $d_i$, the points on $B$, close
  enough together. Moving $c_i$ towards $d_i$ leads to an intersection point
  lying on $B$ in the limit. By continuity, the intersection point can be
  placed close to $B$, see Figure~\ref{fig:order_shor} (right).
\begin{figure}[ht]
\centering
  \includegraphics[width=.9\textwidth]{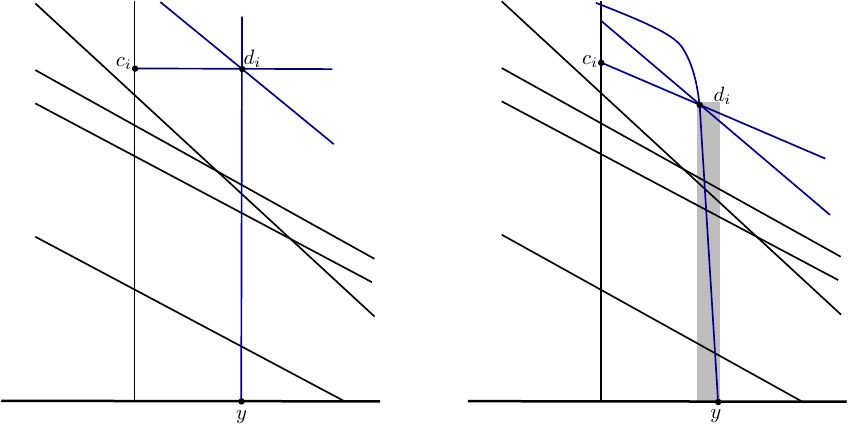}
  \caption{\label{fig:x_order_addition}
  Left: The new intersection points in an addition gadget lie above $y$ if
  $\ell_\infty$ is the line at infinity. Right: The intersection points lie to
  the right of $y$ after perturbing with a projective transformation.}
\end{figure}

  The intersection point of segments in an addition gadget, that does not lie
  on $\ell$ or $\ell_\infty$, lies just to the left of the interval with the old
  segments that lie left of $y$:
  By construction of the gadget the intersection points of the vertical segment
  starting in $y$ has the same $x$-coordinate as $y$, see Figure~\ref{fig:x_order_addition} (left).
  After constructing all gadgets we use a projective transformation to perturb
  the representation, such that the order of the different $x$-coordinates of intersection
  points is preserved.
  For points with the same $x$-coordinates that appear above $y$ in an addition gadget,
  this perturbation will place the points in an interval to the left of $y$ as
  shown in Figure~\ref{fig:x_order_addition} (right).
   
  Now we want to use the fan construction to construct a graph.
  Therefore, we place the point $p$, the origin of the fan, on the coordinates $(0,-C)$ for some large
  negative $C$. If we choose $C$ large enough, then the order of all
  intersection points of segments around $p$ agrees with the $x$-coordinates. 

  Here we have the problem that collinearities between points that lie on
  different segments might occur.

  What we can show this way is a \emph{sandwich version} of the $\exists\mathbb{R}$-completeness:
  Let $H$ be the graph of the  fan we obtain from this construction by assuming no collinearities
  appear between points that do not lie on a common segment or ray.
  Furthermore, we consider all possible combinations of collinearities that can appear in
  the construction and let $F$ be the intersection of all the graph obtained from those fans. 
  Then there is a graph $G$ with $F\subseteq G\subseteq H$ that is a PVG if and only if $V$ is nonempty.
  The possibilities to choose the graph $G$ can be exponentially many.
  Hence for a polynomial reduction we have to determine one of these graphs which we denote by $G_V$,
  that is realizable as PVG in the case that $V$ is nonempty. We do this by
  showing that all unwanted collinearities can be avoided.
 
  Therefore, we choose the positions of the point $a_i$ (and also $b_i$ in a
  multiplication gadget), such that all slopes of segments to $a_i$ and $b_i$ are
  algebraically independent of the coordinates of the points used so far.
  Therefore, we assume the points of the gadgets $g_1,\dots,g_{i-1}$ are already placed.
  After placing the points of the gadget $g_i$ all newly created intersection
  points contain a part of these algebraically independent numbers. Thus the
  lines through two old points cannot go through a new point and vice versa.

  Thus we can avoid all unwanted collinearities.
  We can construct the PVG realization of $G_V$ if and only if $V$ is are
  nonempty, and if we obtain a PVG realization we know that the 
  Now there exists a PVG of $G_V$ realization if $V$ and $V'$ are nonempty.
  The graph $G_V$ can clearly be constructed in polynomial time in the size of
  $V$, this PVG recognition is complete in the existential theory of the reals.
\end{proof}

\begin{remark}
Note that since the construction contains a copy of the order type from Shor's construction, a stronger Mn\"ev-type universality result should hold, 
namely that the realization space of the visibility graph is stably equivalent to the primary semialgebraic set it is constructed from.
 We refer the reader to \cite{orientedMatroids} for more details on the definitions involved.
\end{remark}

\section{Visibility graphs of points on a grid}\label{sec:grid}

From the result of Canny~\cite{canny1988}, we know it is possible to recognize 
point visibility graphs in polynomial space. However, even when the answer to the 
decision problem is positive, it is not clear that the realization of the graph
can be provided as a set of points with integer coordinates. In fact, we 
show in this section that there exist point visibility graphs that cannot be realized 
by points on a grid. We also show that there exist visibility graphs of points on a
grid that require a doubly exponential grid size.

We then turn to the problem of recognizing visibility graphs of points on a grid, and
show that the problem is decidable if and only if the existential theory of the rationals
is decidable, a well-known, major open problem.

\begin{figure}[ht]
\centering
 \includegraphics[width=.3\textwidth]{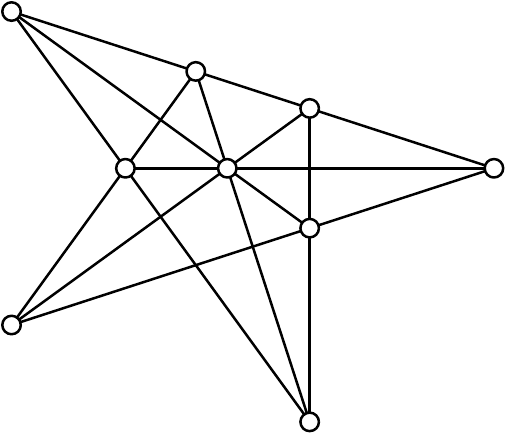}
 \caption{\label{fig:perles} The Perles configuration.}
\end{figure}

\begin{theorem}
There exists a point visibility graph every geometric realization of which has at least one
point with one irrational coordinate.
\end{theorem}

\begin{proof}
We use the so-called {\em Perles configuration} of 9 points on 9 lines illustrated in
Fig.~\ref{fig:perles}. It is known that for every geometric realization of this configuration in the Euclidean
plane, one of the points has an irrational number as one of its coordinate~\cite{G03}. 
We combine this construction with the fan construction described in the previous section.
Hence we pick two lines $\ell$ and $\ell'$ intersecting in a point $p$, such that all lines
of the configuration intersect both $\ell$ and $\ell'$ in the same wedge. Note that up to a 
projective transformation, the point $p$ may be considered to be on the line at infinity and 
$\ell$ and $\ell'$ taken as parallel. We add two non-intersecting segments $s_1$ and $s_2$ close to $p$, 
that do not intersect any line of the configuration. We then shoot a ray from $p$ through
each of the points, and construct the visibility graph of the original points together 
with all the intersections of the rays with the lines and the two segments $s_1, s_2$. 
From Lemma~\ref{lem:uniqueFan}, all the collinearities
of the original configuration are preserved, and every realization of the graph contains a
copy of the Perles configuration.
\end{proof}

Also note that point visibility graphs that can be realized with rational coordinates do not
necessarily admit a realization that can stored in polynomial space in the number of vertices of the graph.
To support this, consider a line arrangement $\mathcal{A}$, and
add a point $p$ in an unbounded face of the arrangement, such that all intersections
of lines are visible in an angle around $p$ that is smaller than $\pi$.
Construct rays $\ell$ and $\ell'$ through the extremal intersection points and $p$.
From Lemma~\ref{lem:uniqueFan}, the fan of this construction gives a PVG that fixes $\mathcal{A}$.
Since there are line arrangements that require integer coordinates of values
$2^{2^{\Theta(|\mathcal A|)}}$~\cite{goodman1990intrinsic} and the fan has 
$\Theta(|\mathcal{A}|^3)$ points we get the following worst-case lower bound on the coordinates of 
points in a realization of a PVG.

\begin{corollary}\label{cor:size_small}
There exists a point visibility graph with $n$ vertices every realization of which
requires coordinates of values $2^{2^{\Theta(\sqrt[3]{n})}}$.
\end{corollary}

We now prove that the recognition problem for visibility graphs on a grid is decidable if and 
only if the \emph{existential theory of the rationals} is decidable. The definition of the latter is analogous to
that of the existential theory of the reals, except we now seek a solution in $\mathbb{Q}^k$.

The computational complexity of answering the question \emph{``Does this object have
a realization on a grid?''} is unknown for various types of objects.
Among those objects are, most prominently, polytopes and
oriented matroids and (non-simple) order types.
Matiyasevich~\cite{matiyasevich1970enumerable} showed that the existential theory of the
integers is undecidable by giving a negative solution to Hilbert's tenth
problem: Deciding whether a diophantine equation has a solution is undecidable.
This cannot be directly applied to a grid realization of a PVG, since a
realization of a PVG with rational coordinates, which can be obtained by a rational solution of
the inequality system, leads to a grid realization by scaling.
Hence for those geometric realizations on the grid the decidability of
Hilbert's tenth problem over the rationals is of interest.
Gr\"unbaum~\cite{grunbaumarrangement} conjectured in 1972 that there is no algorithm that enumerates all
arrangements in the rational projective plane, which is equivalent to the
recognition problem of order types that can be represented on a grid.
This conjecture is still open.

Similar to work of Sturmfels~\cite{sturmfels1987decidability} for oriented
matroids and polytopes we show the following theorem.

\begin{theorem}\label{thm:rationalUndecidable}
The realization problem for visibility graphs of points on a grid is decidable
if and only if the existential theory of the rationals is decidable.
\end{theorem}

Before proving this theorem we point out a connection to finding an upper bound on the grid size
of a PVG that is realizable on a grid.

\begin{corollary}
Suppose the recognition problem for PVG on a grid is undecidable.
Then there is no computable function $f:\mathbb{N}\rightarrow \mathbb{N}$ such
that each PVG with $n$ points that is realizable on a grid
can be drawn on a grid of size $f(n)\times f(n)$.
\end{corollary}

\begin{proof}
We suppose that a computable function $f:\mathbb{N}\rightarrow \mathbb{N}$ exists,
such that every PVG that is realizable on a grid with $n$ vertices can be represented
on a grid of size $f(n)\times f(n)$. Using this function we can give an algorithm that 
decides whether a graph  $G$ with $|V(G)|=n$ has a realization as a PVG on a grid.

We first compute $f(n)$, then for each $x\in [f(n)]^{2n}$ we check whether $G$ is the 
PVG of the point set $\left(x(v_1),y(v_1),\dots,x(v_n),y(v_n)\right)=x$.
If there is such an $x$ the algorithm returns the realization,
otherwise no realization exists.

This algorithm is clearly an effective decision procedure, and thus the recognition
problem for PVG on a grid is decidable -- a contradiction to the assumption.
\end{proof}

\begin{proof}[Proof of Theorem~\ref{thm:rationalUndecidable}]
  To prove this theorem we construct a set of graphs from a semialgebraic set~$V$,
  such that one of the graphs has a rational representation as PVG if and only if~$V$ has.
  As a shortcut we use the result of
  Sturmfels~\cite{sturmfels1987decidability}, that the problem of realizing a
  line arrangement with rational coordinates is complete in the existential
  theory of the rationals.
  The plan is to encode a given (pseudo)line arrangement in a fan by placing
  the origin of the fan $p$, such that all intersection points lie in one
  halfspace.
  By Lemma~\ref{lem:uniqueFan}, a realization of the resulting PVG leads to a
  realization of the line arrangement.
  However, we know neither the radial order of all intersection points around
  $p$ nor which points of the fan, that do not lie on different lines are
  collinear. 
  Thus, we apply the fan construction for all possible radial orderings and all
  possible additional blocked visibilities. This gives a finite set of graphs,
  one of which has a rational PVG realization if and only if the arrangement
  has a rational realization:
  From a rational realization of the line arrangement we obtain a rational PVG
  by applying the fan construction as described with a rational origin $p$.
  This graph is contained in the set of graphs we constructed.
  On the other hand, from a rational PVG representation of one of the graphs we
  obtain a rational PVG representation by Lemma~\ref{lem:uniqueFan}.
\end{proof}

\subsection*{Acknowledgments}
We thank an anonymous SoCG referee for pointing out an error in the original
proof.
We also thank the anonymous D\&CG referee for his numerous suggestions,
which helped to simplify the constructions significantly compared to the
conference version.

\bibliographystyle{plain}
\bibliography{all.bib}

\end{document}